\documentclass{llncs}

\usepackage{amsmath}
\usepackage{graphicx, amssymb}
\usepackage{amsfonts, enumerate}
\usepackage{algorithm,algorithmic}
\usepackage{multirow}

\newcommand{\comment}[1]{}
\newcommand{\runtitle}[1]{{\small \textbf{\boldmath #1}}}

\newcommand{\ind}{{\sc Index}}
\newcommand{\disj}{{\sc Disj}}
\newcommand{\hlis}{{\sc Hidden-IS}}

\newcommand{\bx}{\mathbf{x}}

\newcommand{\ed}{{\rm ed}}
\newcommand{\aed}{\widehat \ed}
\newcommand{\lis}{{\rm lis}}

\newcommand{\sw}{\sigma_w}
\newcommand{\swi}{\sigma_{[i-w+1,i]}}
\newcommand{\eps}{\epsilon}

\begin{document}

\title{
  \mbox{Edit Distance to Monotonicity in Sliding Windows}
}

\author{Ho-Leung Chan\inst{1}
\and Tak-Wah Lam\inst{1}\thanks{T.W. Lam was supported by the GRF Grant HKU-713909E.}
\and Lap-Kei Lee\inst{2}
\and Jiangwei Pan\inst{1}
\and \\
Hing-Fung Ting\inst{1}
\and Qin Zhang\inst{2}
}

\institute{Department of Computer Science, University of Hong Kong, Hong Kong \\
\email{\{hlchan, twlam, jwpan, hfting\}@cs.hku.hk}
\and
MADALGO\footnote{\scriptsize Center for Massive Data Algorithmics -- a Center of the Danish National Research Foundation},
Department of Computer Science, Aarhus University, Denmark \\
\email{\{lklee, qinzhang\}@madalgo.au.dk}
}

\vspace{-.3in}
\maketitle

\vspace{-.2in}
\begin{abstract}
Given a stream of items each associated with a numerical value,
its edit distance to monotonicity is the minimum number of items
to remove so that the remaining items are non-decreasing with respect
to the numerical value. The space complexity of estimating the
edit distance to monotonicity of a data stream is becoming well-understood
over the past few years.
Motivated by applications on network quality monitoring,
we extend the study to estimating the edit distance to monotonicity of
a sliding window covering the $w$ most recent items in the stream for any $w \ge 1$.
We give a deterministic algorithm which can return an estimate within a factor of
$(4+\eps)$ using $O(\frac{1}{\eps^2} \log^2(\eps w))$ space.

We also extend the study in two directions.
First, we consider a stream where
each item is associated with a value from a partial ordered set.
We give a randomized $(4+\epsilon)$-approximate algorithm
using $O(\frac{1}{\epsilon^2} \log \epsilon^2 w \log w)$ space.
Second, we consider an out-of-order stream where each item
is associated with a creation time and a numerical value,
and items may be out of order with respect to their creation times.
The goal is to estimate the edit distance to monotonicity
with respect to the numerical value of items
arranged in the order of creation times.
We show that any randomized constant-approximate algorithm
requires linear space.
\end{abstract}

\section{Introduction}

Estimating the sortedness of a numerical sequence
has found applications in, e.g., sorting algorithms,
database management and webpage ranking (such as Pagerank~\cite{BrP98}).
For example, sorting algorithms can take advantage of
knowing the sortedness of a sequence so as to sort efficiently~\cite{EsW92}.
In relational database, many operations are best performed
when the relations are sorted or nearly sorted over the relevant attributes~\cite{BKF+11}.
Maintaining an estimate on the sortedness of the relations can help determining
whether a given relation is sufficiently nearly-sorted or
a sorting operation on the relation (which is expensive) is needed.
One common measurement of sortedness of a sequence
is its \emph{edit distance to monotonicity} (or ED, in short)~\cite{AJK+02,CMS01,ErJ08,GaG07,GJK+07}:
given a sequence $\sigma$ of $n$ items,
each associated with a value in $[m] = \{1, 2, \dots, m\}$,
the ED of $\sigma$, denoted by $\ed(\sigma)$,
is the minimum number of edit operations required
to transform $\sigma$ to the sequence obtained
by sorting $\sigma$ in non-decreasing order.
Here, an edit operation involves removing an item and re-insert it into a
new position of the sequence.  Equivalently,
$\ed(\sigma)$ is the minimum number of items in $\sigma$ to delete
so that the remaining items have non-decreasing values.
A closely related measurement
is the \emph{length of the longest increasing subsequence} (or LIS) of $\sigma$,
denoted by $\lis(\sigma)$.
It is not hard to see that $\lis(\sigma) = n - \ed(\sigma)$.

With the rapid advance of data collection technologies,
the sequences usually appear in the form of a data stream,
where the stream of items is massive in size (containing possibly billions of items)
and the items are rapidly arriving sequentially.
This gives rise to the problem of estimating ED in the data stream model:
An algorithm is only allowed to scan the sequence sequentially in one pass,
and it also needs to be able to return, at any time,
an estimate on ED of the items arrived so far.
The main concern is the space usage and
update time per item arrival, which, ideally, should
both be significantly smaller than the total data size
(preferably polylogarithmic).

Estimating ED of a data stream is becoming well-understood
over the past few years~\cite{ErJ08,GaG07,GJK+07}.
Gopalan et al.~\cite{GJK+07} showed that computing the ED of a stream
exactly requires $\Omega(n)$ space even for randomized algorithms,
where $n$ is the number of items arrived so far.
They also gave a randomized $(4+\eps)$-approximate algorithm
for estimating ED using space $O(\frac{1}{\eps^2} \log^2 n)$,
where $0 < \eps < 1$.
Later, Ergun and Jowhari~\cite{ErJ08} improved the result
by giving a deterministic $(2+\eps)$-approximate algorithm
using space $O(\frac{1}{\eps^2} \log^2(\eps n))$.
For the closely related LIS problem,
Gopalan et al.~\cite{GJK+07} also gave a deterministic
$(1+\eps)$-approximate algorithm for estimating LIS
using $O(\sqrt{\frac{n}{\eps}})$ space.
This space bound is proven to be optimal in~\cite{GaG07}.

\runtitle{ED in sliding windows.}
The above results consider the sortedness of all items in the stream
arrived so far, which corresponds to the \emph{whole stream model}.
Recently, it is suggested that ED can be an indicator of
network quality~\cite{GKT10}. The items of the stream correspond to
the packets transmitted through a network,
each associated with a sequence number. Ideally, the
packets would arrive in increasing order of the sequence number.
Yet network congestion would result in packet retransmission
and distortion in the packet arrival order, which leads
to a large ED value.
One of the main causes to network congestion is
that traffic is often bursty.
Thus, the network quality can be measured more accurately
if the measurement is based on only recent traffic.
To this end, we propose studying the \emph{sliding window model} where
we estimate the ED of a window covering the latest $w$ items in the stream.
Here $w$ is a positive integer representing the window size.
The sliding window model is no easier than the whole data stream model because
when $w$ is set to be infinity, we need to estimate
ED for all items arrived.

\runtitle{Our results.}
We give a deterministic $(4+\epsilon)$-approximate algorithm
for estimating ED in a sliding window. The space usage is
$O( \frac{1}{\epsilon^2}\log^2(\epsilon w))$, where $w$ is
the window size.
Our algorithm is a generalization
of the algorithm by Gopalan et al. \cite{GJK+07}.
In particular, Gopalan et al. show that ED of the whole stream
can be approximated by the number of ``inverted'' items $j$ such that many items
arrived before $j$ has a value bigger than $j$.
We extend this definition
to the sliding window model. Yet, maintaining the number
of inverted items in a sliding window is non-trivial.
An item $j$ may be inverted when it arrives, but it may
become not inverted due to the expiry of items \emph{arrived earlier}.
We give an interesting algorithm to estimate the number of inverted items
using existing results
on basic counting and quantile estimation over sliding windows.
Our algorithm also incorporates an idea in~\cite{ErJ08}
to remove randomization.

We also consider two extensions of the problem.

{\it $\bullet$ Partial ordered items.}
In some applications, each item arrived is associated
with multiple attributes, e.g., a network packet may contain
both the IP address of the sender and a sequence number.
To measure the network quality,
it is sometimes useful to estimate the \emph{most congested} traffic
coming from a particular sender.
This corresponds to estimating the ED of packets
with respect to sequence number from the same sender IP address.
In this case,
only sequence numbers with the same IP address can be ordered.
We model such a situation by considering items
each associated with a value drawn from a partial ordered universe.
We are interested in estimating the minimum number of
items to delete so that the remaining items are sorted
with respect to the partial order. 
We give a randomized $(4+\epsilon)$-approximate algorithm
using $O( \frac{1}{\epsilon^2} \log \epsilon^2 w \log w)$ space.

{\it $\bullet$ Out-of-order streams.}
When a sender transmits packets to a receiver through
a network, the packets will go through some intermediate routers.
To measure the quality of the route between
the sender and an intermediate router,
it is desirable to estimate the ED of the packets received
by the router from the sender.
Yet in some cases, the router may not be powerful enough to deploy
the algorithm for estimating the ED.
We consider delegating the task of estimation
to the receiver.
To model the situation, whenever a packet arrives,
the intermediate router marks in the packet a timestamp recording
the number of packets~received~thus far
(which can be done by maintaining a single counter).
Hence, when the packets
arrive at the receiver, each packet has both a sequence number
assigned by the sender and a timestamp marked by the router.
Note that the packets arrived at the receiver may be out-of-order
with respect to the timestamp.
Such stream corresponds to an \emph{out-of-order~stream}.

To measure the
network quality between the sender and the router,
the receiver can estimate the ED with respect to the sequence number
when the items are arranged in increasing order of the timestamps.
Intuitively, the problem is difficult as items can
be inserted in arbitrary positions of the sequence according
to the timestamp.
We show strong space lower bounds even in the whole stream model.
In particular, any randomized constant-approximate
algorithm for estimating ED of an out-of-order stream requires $\Omega( n )$ space,
where $n$ is the number of items arrived so far.
An identical lower bound holds for estimating the LIS.
Like most streaming lower bounds, our lower bounds are proved based on
reductions from two communication problems, namely, the \ind~problem and the
\disj~problem.  Optimal communication lower bounds for randomized protocols are
known for both problems \cite{Abl96,Jay09}.

\comment{
\runtitle{Techniques.}
For the problem of estimating ED in the whole data stream, existing algorithms try to identify a
subset of ``bad'' items, who has a significant effect on the sortedness of the stream.
At the same time, the
size of the subset is a good approximation to the actual ED of the stream.  Therefore, the algorithm
only need to maintain a counter for the number of ``bad'' items, which will never decrease.  However,
for the sliding window version of this problem, as the window slides, ``bad'' items may eventually
become ordinary items even if they have not expired themselves. So one counter will not
work in this case.  We
deal with this situation by modeling those ``bad'' items as an out-of-order substream.  By employing
a basic counting structure for out-of-order streams, the size of the ``bad'' subset can be estimated
dynamically.

For lower bounds in out-of-order streams, we use different methods for ED and LIS.  For
ED, we construct a large set of input streams such that a correct
algorithm needs to differentiate any two of them (i.e. different memory states). Therefore, the
algorithm needs to use space the logarithm of the total number of input streams.  On the other hand,
we prove the lower bound for LIS through a reduction from a communication
problem, whose communication lower bound is known.  Both bounds are deterministic and linear in the
size of the stream for constant approximation.
}

\runtitle{Organization.}
Section~\ref{sec:definitions} and~\ref{sec:main} give the formal problem definitions
and our main algorithm for
estimating ED, respectively.
Section~\ref{sec:extensions} considers out-of-order streams.
Due to the page limit, extension to
partial ordered items is left to the full paper.

\section{Formal problem definitions}\label{sec:definitions}

\runtitle{Sortedness of a stream.}
Consider a stream $\sigma$ of $n$ items,
$\langle \sigma(1), \sigma(2), \dots, \sigma(n) \rangle$
where each $\sigma(i)$ is drawn
from $[m] = \{1, 2, \dots, m\}$.
The \emph{edit distance to monotonicity} (ED) of $\sigma$,
denoted by $\ed(\sigma)$,
is the minimum number of items required
to remove so as to obtain an increasing subsequence of $\sigma$,
i.e., $\langle \sigma(i_1), \sigma(i_2), \dots, \sigma(i_k) \rangle$
such that $\sigma(i_1) \le \sigma(i_2) \le \cdots \le \sigma(i_k)$
for some $1 \le i_1 < i_2 < \cdots < i_k \le n$.
We use $\lis(\sigma)$ to denote the
\emph{length of the longest increasing subsequence} (LIS) of $\sigma$.
Note that $\lis(\sigma) = n - \ed(\sigma)$.
The sortedness can be computed based
on the \emph{whole stream} (all items in $\sigma$ received thus far)
or a \emph{sliding window} covering the most recent $w$ items,
denoted by $\sw$, for $w \ge 1$.
Note that the whole stream model can be viewed as
a special case of the sliding window model with window size $w = \infty$.
A streaming algorithm has only limited space
and can only maintain an estimate on
the sortedness of $\sw$.
For any $r \ge 1$,
a $r$-approximate algorithm for estimating $\ed(\sw)$
returns, at any time,
an estimate $\aed(\sw)$
such that $\ed(\sw) \le \aed(\sw) \le r \cdot \ed(\sw)$.
We can define a $r$-approximate algorithm
for estimating $\lis(\sw)$ similarly.

\runtitle{Partial ordered universe.}
We also consider a partial ordered universe
with binary relation $\preceq$.
A subsequence of $\sigma$ with length $\ell$,
$\langle \sigma(i_1), \sigma(i_2), \cdots, \sigma(i_\ell) \rangle$,
is increasing
if for any $k \in [\ell-1]$, $\sigma(i_k) \preceq \sigma(i_{k+1})$.
Then for any window size $w \ge 1$,
$\ed(\sw)$ and $\lis(\sw)$ can be defined analogously as before.

\runtitle{Out-of-order stream.}
The data stream described above is an \emph{in-order} stream,
which assumes items arriving in the same order as their creation time.
In an \emph{out-of-order stream},
each item is associated with a distinct integral time-stamp recording
its creation time, which may be different from its arrival time.
Precisely, an out-of-order stream $\sigma$
is a sequence of tuples $\langle t_i, v_i\rangle$ ($i \in [n]$)
where $t_i$ and $v_i$ are the timestamp and value of the $i$-th item.
The sortedness of $\sigma$ is defined based on the permuted sequence
$V(\sigma) = \langle v_{i_1}, v_{i_2}, \dots, v_{i_n} \rangle$
such that $t_{i_1} \le t_{i_2} \le \cdots \le t_{i_n}$, i.e.,
$\ed(\sigma) := \ed(V(\sigma))$ and $\lis(\sigma) := \lis(V(\sigma))$.

\section{A $(4+\epsilon)$-approximate algorithm for estimating ED}\label{sec:main}

In this section, we consider a stream $\sigma$ of items with values drawn
from a set $[m] = \{1, 2, \dots, m\}$,
and we are interested in estimating the ED of a sliding window covering
the most recent $w$ items in~$\sigma$.
We give a deterministic $(4+\epsilon)$-approximate algorithm
which uses $O( \frac{1}{\epsilon^2}\log^2(\epsilon w ))$ space.

Our algorithm is based on an estimator $R(i)$, which is a generalization of the estimator in~\cite{GJK+07}
to the sliding window model.
Let $i$ be the index of the latest arrived item.
The sliding window we consider is
$\swi = \langle \sigma(i-w+1), \sigma(i-w+2), \dots, \sigma(i) \rangle$.
For any item $\sigma(j)$, let $inv(j)$ be the set of items
arrived before $\sigma(j)$ but have greater values than
$\sigma(j)$, i.e., $inv(j) = \{ k : k < j \mbox{ and } \sigma(k) > \sigma(j)\}$.
We define an estimator $R(i)$ for $\ed(\swi)$ as follows.

\begin{definition}
Consider the current sliding window $\swi$.
We define $R(i)$ to be the set of indices $j \in [i-w+1, i]$ such that
there exists $k \in [i-w+1, j-1]$ with $|[k, j-1] \cap inv(j)|
> \frac{j-k}{2}$.
\end{definition}


\begin{lemma}[\cite{GJK+07}] \label{lem:estimator}
$\ed( \sigma_{[i-w+1,i]}) /2 \le |R(i)| \le 2 \cdot \ed( \sigma_{[i-w+1,i]})$.
\end{lemma}

Hence, if we know $|R(i)|$,
we can return $2 |R(i)|$ as an estimation for $\ed(\sigma_{[i-w+1,i]})$
and it gives a 4-approximation algorithm.
However, maintaining $R(i)$ exactly requires space linear
to the window size. In the following, we show how to approximate
$R(i)$ using significantly less space.

\subsection{Estimating $R(i)$}
We first present our algorithm and then show that it
can approximate $R(i)$. Our algorithm will
make use of two data structures.
Let $\eps'$ be a constant in $(0,1)$ (which will be
set to $\epsilon/35$ later).

{\bf $\epsilon'$-approximate quantile data structure
\boldmath $\mathcal{Q}$}:
Let $Q$ be a set of items. The rank of an item in~$Q$
is its position in the list formed by sorting $Q$
from the smallest to the biggest. For any $\phi \in [0,1]$,
the $\epsilon'$-approximate $\phi$-quantile of $Q$
is an item with rank in $[(\phi-\epsilon')|Q|,(\phi+\epsilon')|Q|]$.
We maintain an $\epsilon'$-approximate $\phi$-quantile
data structure given in~\cite{LLX+04} which
can return, at any time, an $\epsilon'$-approximate
$\phi$-quantile of the most recent $w'$ items
for any $w'\le w$.
This data structure takes
$O(\frac{1}{(\epsilon')^2}\log ^2 (\epsilon' w))$ space.

{\bf $\epsilon'$-approximate basic counting data structure \boldmath$\mathcal{B}$}:
When an item $\sigma(i)$ arrives, we may
associate a token with some item $\sigma(k)$
where $k < i$. The association is permanent and
an item may be associated with more than one token.
At any time, we are interested in the number
of tokens associated with the most recent $w$ items.
We view it as a stream $\sigma_{\rm token}$ of tokens,
each of which has a timestamp $k$ if it is associated to $\sigma(k)$,
and we want to return the number of tokens with timestamp in $[i-w+1, i]$.
Note that the tokens may be out-of-order with respect to the timestamp,
leading to the basic counting problem for out-of-order stream
considered in~\cite{CKT08}.
We maintain their $\epsilon'$-approximate basic counting data structure
on $\sigma_{\rm token}$ which can return, at any time, an estimate $\hat t$
such that $| \hat t - t| \le \epsilon' t$, where
$t$ is the number of tokens associated with
the latest $w$ items.
It takes
$O( \frac{1}{\epsilon'} \log w \log(\frac{\epsilon' B}{\log w}))$ space,
where $B$ is the maximum number of tokens associated within
any window of $w$ items.
As we may associate one token upon any
item arrival, $B$ is at most~$w$. 

We are now ready to define our algorithm, as follows.
\begin{center}
\vspace{-.1in}
  \begin{tabular}{l}
    \hline
    {\bf Algorithm 1.} Estimating ED in sliding windows\\
    \hline
    {\bf Item arrival:} Upon the arrival of item $\sigma(i)$, do \\
    ~ ~ For $k=i-1,i-2,\cdots,i-w+1$\\
    ~ ~ ~ ~ ~Query $\mathcal{Q}$ for the $(\frac{1}{2}-\epsilon')$-quantile of $\sigma_{[k, i-1]}$.
             Let $a$ be the returned value.\\
    ~ ~ ~ ~ ~If $a>\sigma(i)$,
    associate a token to $\sigma(k)$, i.e., add an item with timestamp $k$\\
    ~ ~ ~ ~ ~to the stream $\sigma_{\rm token}$.
                  Break the for loop. \vspace{.1in}\\
   {\bf Query:}
   Query $\mathcal{B}$ on the stream $\sigma_{\rm token}$ for the number of tokens associated
   with\\
   the last $w$ items and
   let $\hat t$ be the returned answer.
   Return $\hat t / (\frac{1}{2}-2\epsilon')( 1-\epsilon')$ as\\
   the estimation $\aed( \sigma_{[i-w+1,i]})$.   \\
    \hline
  \end{tabular}
\end{center}

Let $R'(i)$ be
the set of indices $j$ such that
when $\sigma(j)$ arrives, we associate a token
to an item $\sigma(k)$ where $k \in [i-w+1, i]$.
Observe that $R'(i)$ is an approximation of $R(i)$
in the following sense.

\begin{lemma}\label{lem:R'}
$R'(i)$ contains \emph{all} indices $j \in [i-w+1, i]$ satisfying
that there exists $k \in [i-w+1, j-1]$ such that
$|[k, j-1]\cap inv(j)| > (\frac{1}{2} + 2\epsilon')(j-k)$.
Furthermore, \emph{all} indices $j$ contained in
$R'(i)$ must satisfy that
there exists $k \in [i-w+1, j-1]$ such that
$|[k, j-1]\cap inv(j)| > \frac{j-k}{2}$.
\end{lemma}

\begin{proof}
An index $j$ is in $R'(i)$ if
$\sigma(j) < a$ when $\sigma(j)$ arrives, where
$a$ is the $\epsilon'$-approximate $(\frac{1}{2} - \epsilon')$-quantile
for some interval $\sigma_{[k, j-1]}$. Note that the rank of $a$
in $\sigma_{[k,j-1]}$ is at least $(\frac{1}{2} - 2\epsilon')(j-k)$.
Therefore, if $|[k, j-1]\cap inv(j)| > (\frac{1}{2} + 2\epsilon')(j-k)$,
the rank of $\sigma(j)$ is less than $(j-k) - (\frac{1}{2} + 2\epsilon')(j-k)
= (\frac{1}{2} - 2\epsilon')(j-k)$,
so $\sigma(j) < a$ and $j$ must be included in $R'(i)$.
On the other hand,
the rank of $a$ in $\sigma_{[k,j-1]}$ is at most $\frac{j-k}{2}$.
Since $a > \sigma(j)$, we conclude that
all indices $j \in R'(i)$ satisfy $|[k, j-1]\cap inv(j)| > \frac{j-k}{2}$.
\hfill\qed
\end{proof}

We show that $|R'(i)|$ is a good approximation
for $\ed(\swi)$, as follows.

\begin{lemma} \label{lem:approx-estimator}
$(\frac{1}{2}-2\epsilon')\cdot \ed(\swi) \le |R'(i)|
\le 2 \cdot \ed(\swi)$.
\end{lemma}
\begin{proof}
  We observe that by Lemma~\ref{lem:R'},
  any index $j$ in $R'(i)$ must be also
  in $R(i)$. Hence, $R'(i) \subseteq R(i)$ and
  $|R'(i)| \le |R(i)| \le 2 \cdot \ed(\swi)$
  (by Lemma~\ref{lem:estimator}).

  Now, we show $(\frac{1}{2} - 2\epsilon')\cdot \ed(\swi) \le |R'(i)|$ by giving an iterative pruning procedure to obtain an
  increasing subsequence (may not be the longest).  First let
  $x = i+1$ and
  $\sigma(x) = \infty$.  Find the largest $j$ such that $i-w+1\le j<x$ and
  $j\notin R'(i) \cup inv(x)$ and delete the interval $[j+1,x-1]$.
  We then let $x=j$ and repeat the process
  until no such $j$ is found.
  As each~$x$ is not in $R'(i)$,
  Lemma~\ref{lem:R'} implies that in every interval
  that we delete, the fraction of items of $R'(i)$ is at
  least $(\frac{1}{2}-2\epsilon')$.
  Note that eventually all items in $R'(i)$ will be deleted.
  Thus, $|R'(i)|
  \ge (\frac{1}{2}-2\epsilon')\cdot \mbox{(number of deleted items)} \ge
  (\frac{1}{2}-2\epsilon')\cdot \ed(\swi)$.
\hfill\qed
\end{proof}

Note that $|R'(i)|$ equals the number of tokens associated
with the most recent $w$ items. Since $\cal B$ is only an $\eps'$-approximate
data structure, the value $\hat t$ returned only satisfies that
$(1-\epsilon')|R'(i)| \le \hat t \le (1+\epsilon')|R'(i)|$.
Since we report $\aed(\swi) = \hat t / (\frac{1}{2}-2\epsilon')( 1-\epsilon')$ as
the estimation, we conclude with the following approximation ratio.

\begin{lemma}
$\ed(\swi) \le \aed(\swi)
\le \frac{2(1+\epsilon')}{(1/2-2\epsilon')(1-\epsilon')} \cdot \ed(\swi)$
\end{lemma}

For any $\epsilon \le 1$, we can set $\epsilon' = \epsilon/35$.
Then, $\frac{2(1+\epsilon')}{(1/2-2\epsilon')(1-\epsilon')} \cdot \ed(\swi) \le (4+\epsilon)\cdot \ed(\swi)$.
The total space usage of the two data structures is
$O(\frac{1}{\eps^2}\log ^2(\eps w) + \frac{1}{\eps} \log w \log(\eps w))$.
If $\eps > \frac{1}{w}$, $\log w = O(\frac{1}{\eps} \log (\eps w))$
and thus the total space usage is $O(\frac{1}{\eps^2}\log ^2(\eps w))$.
Otherwise, we can store all items in the window, which
only requires $O(w) = O(\frac{1}{\eps})$ space.

\vspace{.1in}
\runtitle{Improving the running time.}  The per-item update time of the algorithm is $O(w)$
because the algorithm checks the interval $I = [k, i-1]$ for every length $|I| \in [w-1]$.
An observation in~\cite{ErJ08} is that an $\frac{\epsilon'}{2}$-approximate
$\phi$-quantile of an interval with length $|I|$ is also an $\epsilon'$-approximate $\phi$-quantile
for all intervals with length $|I|+1, \cdots, (1+\frac{\epsilon'}{2})|I|$. Hence we only need to check
$O(\frac{1}{\epsilon'} \log w)$ intervals of length $1,2,\cdots,(1+\frac{\epsilon'}{2})^i,(1+\frac{\epsilon'}{2})^{i+1},\cdots,w$.
Then we obtain an
$\epsilon'$-approximate quantile for every interval.
Note that the query time for returning an approximate quantile
is $O(\frac{1}{\epsilon'} \log^2 w)$,
and the per-item update time of the two data structures is
$O(\frac{1}{\epsilon^2}\log ^3 w)$~\cite{CKT08,LLX+04}.
We conclude with the main result of this section.


\begin{theorem}
There is a deterministic $(4+\epsilon)$-approximate algorithm
for estimating ED in a sliding window of the latest $w$ items.
The space usage is $O(\frac{1}{\epsilon^2}\log ^2(\epsilon w))$
and the per-item update time is
$O(\frac{1}{\epsilon^2}\log ^3 w)$.
\end{theorem}

\runtitle{Remark.} For the whole stream model, the state-of-the-art result is a
$(2+\epsilon)$-approximation in \cite{ErJ08}.  They gave an improved estimator
$R(i)$ as the set
of indices $j$ such that there exists $k < j$ with $|[k,j-1]\cap inv(j)| >
|[k,j-1]\cap R(i)|$. In other words, whether an index belongs to $R(i)$ or not
depends on the number of members of $R(i)$ before that index.  Note that
a member of $R(i)$ could become a nonmember due to window expiration.  Therefore, an
index $j$ that is not a member of $R(i)$ initially, may later become a member
if some of the previous $R(i)$ members become nonmembers.  This makes estimating
this improved $R(i)$ difficult in the sliding window model.

\section{Lower bounds for out-of-order streams}\label{sec:extensions}

In this section, we consider an out-of-order stream $\sigma$
consisting of a sequence of items $\sigma(i) = \langle t_i, v_i\rangle$ for $i\in[N]$,
where $t_i$ and $v_i$ are the timestamp and
value of the $i$-th item, respectively.
Recall that the sortedness of the stream is
measured on the derived value sequence by rearranging the items in
non-decreasing order of the timestamps. 
We show that even for the whole data stream model,
any randomized constant-approximate algorithm
for estimating ED or LIS requires $\Omega(N)$ space. 
In fact, a stronger lower bound holds for ED: any randomized algorithm that decides 
whether ED equals 0 uses $\Omega(N)$ space.  Our proofs follow from reductions from
two different communication problems.

\subsection{Estimating ED in an out-of-order stream}

\begin{theorem}\label{thm:lb_ed}
  Consider an out-of-order stream $\sigma$ of size $N$.
  Any randomized algorithm that distinguish between the cases that 
  $\ed(\sigma) = 0$ and that $\ed(\sigma) \ge 1$ must use $\Omega(N)$ bits.
  Therefore, for arbitrary constant $r \ge 1$, any randomized $r$-approximation to
  $\ed(\sigma)$ requires $\Omega(N)$ bits.
\end{theorem}

We prove the above lower bound by showing a reduction from the classical
communication problem \ind, which has strong communication lower bound. 

The problem \ind$(x,i)$~is a two-player one-way communication game.  Alice holds a binary
string $x \in \{0,1\}^n$ and Bob holds an index $i \in [n]$.
In this communication game, Alice sends one message to Bob and Bob is
required to output the $i$-th bit of $x$, i.e.  $x_i$, based on 
the message received.  A trivial protocol is for Alice to send all her
input string $x$ to Bob, which has communication complexity of $n$ bits.  It turns
out that this protocol is optimal.  Particularly, Alice must
communicate $\Omega(n)$ bits in any randomized protocol for \ind
~\cite{Abl96}.

\begin{proof}[of Theorem~\ref{thm:lb_ed}]
  Given an out-of-order stream with length $N$, suppose there is a
  randomized algorithm $\mathcal{A}$ that can determine whether its ED equals to
  0 or is at least 1 using $S$ memory bits.  We define a randomized
  protocol $\mathcal{P}$ for \ind$(x,i)$ for $n = N - 1$: Alice constructs
  (hypothetically) an out-of-order stream $\sigma$ with length $n$ by setting
  \begin{equation}
    \sigma(j) = \left\{ \begin{array}{ll}
    \langle 2j-1, 3j-2\rangle, &\textrm{ if } x_j = 0 \\
    \langle 2j-1, 3j\rangle, &\textrm{ if } x_j = 1 
  \end{array}\right.
  \label{eq:ed_map}
  \end{equation}
  Alice then simulates algorithm $\mathcal{A}$ on stream $\sigma$ and sends the
  content of the working memory to Bob.  Bob constructs another stream item
  $\sigma(n+1) = \langle 2i, 3i-1 \rangle$ to continue running algorithm
  $\mathcal{A}$ on it and obtains the output.  If the output says $\ed(\sigma) =
  0$, Bob outputs 0; otherwise, Bob outputs 1.

  It is not hard to see that \ind$(x,i) = x_i = 0$ implies $\ed(\sigma) = 0$ and
  \ind$(x,i) = 1$ implies $\ed(\sigma) = 1$.  Therefore, if algorithm
  $\mathcal{A}$ reports the correct answer with high probability, the protocol
  $\mathcal{P}$ outputs correctly with high probability, and thus is a valid
  randomized protocol for \ind.  In the protocol, the number of bits communicated by Alice is
  at most $S$.  Combining the $\Omega(n) = \Omega(N)$ lower bound, we obtain
  that $S = \Omega(N)$, completing the proof.
\hfill\qed
\end{proof}

\subsection{Estimating LIS in an out-of-order stream}

\begin{theorem}\label{thm:lb_lis}
  Consider an out-of-order stream $\sigma$ with size $N$.  Any randomized
  algorithm that outputs an $r$-approximation on $\lis(\sigma)$ must use
  $\Omega(N/r^2)$ bits.
\end{theorem}

\begin{proof}
We prove the lower bound by considering the $t$-party set disjointness
problem \disj.  The input to this communication game is a binary $t \times \ell$
matrix $\bx \in \{0, 1\}^{t\ell}$, and each player $P_i$ holds one row of
$\bx$, the 1-entries of which indicate a subset $A_i$ of $[\ell]$.  The input
$\bx$ is called \emph{disjoint} if the $t$ subsets are pairwise
disjoint, i.e., each column of $\bx$ contains at most one 1-entry; and it is
called \emph{uniquely intersecting} if the subsets $A_i$ share a unique common
element $y$ and the sets $A_i - \{y\}$ are pairwise disjoint, meaning that in
$\bx$, except one column with entries all equal to 1, all the other columns have at most one
1-entry.  The objective of the game is to distinguish between the two types of
inputs.  To obtain the space lower bound, we only need to consider a restricted 
version of \disj~where, according to some probabilistic protocol, the first
$t-1$ players in turn send a message privately to his next neighboring player and the last
player $P_t$ outputs the answer.

An optimal lower bound of $\Omega(\ell/t)$ total communication is known for
\disj~even for general randomized protocols (with constant success probability) 
\cite{Jay09}, and thus the lower
bound also holds for our restrited one-way private communication model.  By
giving a reduction and setting the parameters appropriately, we can obtain the
space lower bound.

Given a randomized algorithm that outputs $r$-approximation to the LIS of any out-of-order stream
with length $N$, using $S$ memory bits, we
define a simple randomized protocol for \disj~for $t = 2r = o(N)$ and $\ell =
N+1-t = \Theta(N)$.  Let $\bx$ be the input $t\times \ell$ matrix.  The first player
$P_1$ creates an out-of-order stream $\sigma$ by going through his row of input
$R_1(\bx)$ and inserting a new item $\langle (j-1)t+1, (\ell-j)t+1 \rangle$ to
the end of the stream
whenever an entry $x_{1j}$ equals to 1.  He then runs the streaming algorithm on
$\sigma$ and sends the content of the memory to the second player.  In general,
player $P_i$ appends a new item $\langle (j-1)t+i, (\ell-j)t+i \rangle$ to the
stream for each nonzero entry $x_{ij}$, simulates the streaming algorithm and
communicates the updated memory state to the next player.  Finally, player $P_t$ obtains
the approximated LIS of stream $\sigma$.  If it is at most $r$ he reports
that the input $\bx$ is disjoint; else, he reports it is uniquely intersecting.
It's easy to verify that if the input $\bx$ is disjoint, the correct LIS of stream $\sigma$ is 1,
while if it is uniquely intersecting, the correct LIS of $\sigma$ is $t$.
Consequently, if the streaming algorithm outputs an $r$-approximation to
$\lis(\sigma)$ with probability at least $2/3$, the protocol for \disj~is
correct with constant probability, using total communication at most $(t-1)S$.
Following the lower bound for \disj, this implies $(t-1)S \ge \Omega(\ell/t)$,
i.e., $S = \Omega(\ell/t^2) = \Omega(N/r^2)$.  Theorem~\ref{thm:lb_lis} follows.
\end{proof}

\runtitle{Remark.} Actually, for deterministic algorithms, we can obtain a
slightly stronger lower bound of $\Omega(N/r)$ for $r$-approximation, by
a reduction from the \hlis~problem used in \cite{GaG07} to prove the
$\Omega(\sqrt{N})$ lower bound for approximating LIS of an in-order stream.  The
reduction is similar to the above, and if we set the approximation ratio
$r$ to a constant, the lower bounds become linear in both cases.  Therefore, we
neglect the details here.

\section*{Acknowledgement}
We thank the anonymous reviewers for helpful comments and for pointing out the
randomized lower bounds to us.

\end{document}